\documentclass[12pt,english]{article}
\usepackage{lmodern}
\usepackage[T1]{fontenc}
\usepackage[latin9]{inputenc}
\usepackage[letterpaper]{geometry}
\usepackage{color}
\usepackage{babel}
\usepackage{amsmath}
\usepackage{amsthm}
\usepackage{amssymb}
\usepackage{stmaryrd}
\usepackage{setspace}
\usepackage[authoryear]{natbib}
\usepackage[pdfusetitle,
 bookmarks=true,bookmarksnumbered=true,bookmarksopen=false,
 breaklinks=false,pdfborder={0 0 1},backref=false,colorlinks=true]
 {hyperref}
\hypersetup{
 colorlinks,pagebackref=false,citecolor=dark-blue,urlcolor=black,linkcolor=dark-blue}

\makeatletter
\theoremstyle{plain}
\newtheorem{prop}{\protect\propositionname}
\theoremstyle{plain}
\newtheorem{cor}{\protect\corollaryname}
\theoremstyle{plain}
\newtheorem{thm}{\protect\theoremname}
\theoremstyle{plain}
\newtheorem{assumption}{\protect\assumptionname}
\theoremstyle{plain}
\newtheorem{lem}{\protect\lemmaname}

\date{}

\usepackage{datetime}
\newdateformat{mydate}{\monthname[\THEMONTH] \THEYEAR}

\usepackage{amsfonts}
\usepackage{amsthm}
\usepackage{afterpage} 
\usepackage{color}
\usepackage{epsf}

\usepackage{fullpage}
\usepackage{latexsym}
\usepackage{bm}
\usepackage{bbm}

\usepackage[multiple]{footmisc}

\usepackage{floatrow}
\usepackage[margin=10pt,labelfont=bf,tableposition=bottom ]{caption}

\usepackage{tikz}
\usepackage{pgfplots}
\usetikzlibrary{decorations.pathreplacing,calligraphy}

\definecolor{dark-red}{rgb}{0.4,0.15,0.15}
\definecolor{dark-blue}{rgb}{0.15,0.15,0.75}
\definecolor{medium-blue}{rgb}{0,0,0.5}

\usetikzlibrary{arrows.meta,positioning}
\usepackage{amsmath,amssymb}

\makeatother

\providecommand{\assumptionname}{Assumption}
\providecommand{\corollaryname}{Corollary}
\providecommand{\lemmaname}{Lemma}
\providecommand{\propositionname}{Proposition}
\providecommand{\theoremname}{Theorem}

\begin{document}
\title{Profiling and Endogenous Valuation}
\author{Anh Nguyen\thanks{Tepper School of Business, Carnegie Mellon University. Email: anhnguyen@cmu.edu.}
\and  Teck Yong Tan\thanks{Simon Business School, University of Rochester. Email: t.tan@rochester.edu.}}
\date{\mydate\today\\}
\maketitle
\begin{abstract}
We study a monopolist facing a buyer whose valuation is determined
by pre-trade investment. Before setting price, the seller observes
a signal about the buyer's private investment cost (buyer profiling).
Information that helps the seller extract surplus can also undermine
the buyer's incentive to create it. We characterize the buyer-seller
payoffs attainable across all possible profiling. On the Pareto frontier,
if investment increases, hold-up risk always raises the payoff that
the seller captures faster than the surplus that the investment creates.
Protecting buyer welfare therefore requires discouraging investment,
even though investment is socially efficient.
\end{abstract}



\newpage

\section{\protect\label{Section:=000020Introduction}Introduction}

Information is central to price discrimination. A seller who knows
more about a buyer can better predict his willingness to pay and tailor
prices. As firms increasingly collect, purchase, and analyze consumer
data, a natural question is how such buyer profiling affects market
outcomes and welfare.

\citet{Bergemann_Brooks_Morris_2015(AER)} (hereafter BBM) and a growing
literature study this question when buyer valuations are \emph{fixed}.
In many markets, however, valuations depend on costly actions that
the buyer takes before trade. A consumer may learn a firm's product
ecosystem or migrate his data and devices into it; a business customer
may train employees and adapt internal processes around a vendor's
platform. Such actions raise the value of trade, but the incentive
to take them depends on the price the buyer expects to face. Information
that helps the seller extract surplus can therefore undermine the
buyer's incentive to create it. Thus, with endogenous valuation, information
affects \emph{both} surplus creation and surplus division.

We study a monopolist facing a buyer with endogenous valuation and
analyze the welfare implications of the seller having information
about the buyer. The buyer (he) has initial valuation $L$ for the
seller\textquoteright s (she) product but can pay a privately known
cost $c$ to raise his valuation to $H>L$ before they meet.\footnote{Assuming binary investment simplifies the exposition; section \ref{Section:Multiple-Investments}
shows that our results extend to multiple investment options. Such
binary ``all-or-nothing'' investment also fits our motivating settings,
where what matters is whether the buyer crosses a qualitative threshold---becoming
fluent in a product ecosystem, or fully adapting operations around
a vendor---rather than by how much.} Investing is socially efficient for every type. The seller does not
observe whether the buyer invested, but before setting price she observes
a signal about $c$, interpreted as buyer profiling. Following BBM's
approach, we ask which buyer-seller payoff pairs can arise across
all possible profiling.

With fixed valuations, BBM shows that any individually rational division
of the available surplus---including the first-best surplus---is
attainable. Endogenous valuation generates two immediate departures.
First, the first-best surplus is unattainable: it requires every type
to invest, but if every type invests, the seller would optimally price
at $H$, and anticipating this, no type would invest. Second, more
information can strictly hurt the seller. If the seller perfectly
observes $c$, her payoff is her \emph{lowest} individually rational
payoff. In both departures, the driving force is that the seller's
information determines not only her ability to extract surplus, but
also the buyer's incentive to create it.

We begin the analysis by identifying two simple bounds on attainable
payoffs. Besides individual rationality, an attainable payoff pair
must satisfy two constraints. The first is a \emph{surplus constraint}:
there must exist some investment level that generates enough surplus
to support the payoffs. The second is novel to endogenous valuation:
the seller's payoff must be at least what she could earn by ignoring
her signal and always charging $H$. We call this the \emph{hold-up
risk constraint}. It imposes an upper bound on investment: more investment
creates more high-valuation buyers and makes pricing high more tempting
to the seller. If this temptation is too strong, the buyer\textquoteright s
investment incentive unravels, reminiscent of classic hold-up. Proposition
\ref{Proposition:AC} shows that these two necessary constraints are
also sufficient for payoff attainability. Thus, as in BBM, attainable
payoffs are governed by simple necessary bounds, but the new constraint
changes the geometry of the payoff set, which is characterized in
Theorem \ref{Theorem:Characterization} and Figure \ref{Figure:Attainable-Payoffs}. 

Our central result is that the hold-up risk constraint restricts not
only surplus creation but also surplus division. Proposition \ref{Proposition:Pareto-frontier}
shows that, on the Pareto frontier of the attainable set, increasing
the buyer\textquoteright s payoff requires lowering investment and,
with it, total welfare. This is because of a subtle property of the
hold-up risk constraint: as equilibrium investment expands, the payoff
that the seller must capture grows \emph{faster} than the surplus
that the investment creates. Sustaining more investment therefore
requires giving the seller more than the additional surplus created.
This differs from fixed-valuation price discrimination, where information
can freely divide the fixed surplus; here, surplus division and creation
are inseparable, and favoring the buyer on the frontier is possible
only by creating less. Protecting buyer welfare can therefore require
information that sustains \emph{less} investment, even though investment
is always socially efficient. 

In Proposition \ref{Proposition:Construction}, we show what such
information looks like. Every attainable payoff can be attained by
a dual-cutoff equilibrium: one cutoff determines the surplus creation
margin, and the other determines the surplus division margin. More
importantly, the construction identifies which information matters
for whom. Information about high-cost types, who are unlikely to invest,
benefits both parties because it prevents inefficient trade breakdown
and expands the surplus that can be sustained. Information about low-cost
types, who are likely to invest, is instead distributive: it helps
the seller identify the surplus created but weakens the buyer\textquoteright s
incentive to create it. The policy implication for any welfare objective
is therefore not whether profiling should be uniformly made more or
less accurate, but whose information should be.

\paragraph*{Related Literature. }

This paper contributes to the literature on the welfare implications
of the seller's information about buyers, beginning with BBM, and
extended to, among others, a multi-product seller (\citealp{Haghpanah_Siegel_2022(AERI),Haghpanah_Siegel_2023(JPE),BergemannHeumannWang2026}),
downward-sloping demand (\citealp{Farboodi_Haghpanah_Shourideh_2025}),
correlated valuation (\citealp{Kartik_Zhong_2025(WP)}), a privately
informed seller (\citealp{Ichihashi_Smolin_2025(GEB)}), and seller
competition (\citealp{Bergemann_Brooks_Morris_2025(AEJ),Elliott_et_at_2025}).
In this literature, valuations are fixed, so information affects only
how surplus is divided. Here, the surplus itself is endogenous to
the information structure, which entangles the division of surplus
with its creation. \citet{Piccione_Tan_1996(IER)}, \citet{Tan_1996(CJE)},
and \citet{Gershkov_et_al_2021(JPE)} study endogenous valuations
in auctions, analyzing how the selling mechanism affects investment
incentives; we instead fix the selling mechanism and consider how
the seller's information affects investment incentives. Our analysis
is also broadly related to Bayesian persuasion \citep{Kamenica_Gentzkow_2011(AER)}.
The difference here is that the privately informed player (the buyer)
also takes an action that affects the ``receiver's'' (the seller's)
payoff. \citet{Ichihashi_2025(JET)} studies a problem with a similar
structure, analyzing information and crime deterrence under an enforcement
resource constraint. 

\section{\protect\label{Section:Model}Model}

\paragraph{Primitives.}

A risk-neutral buyer (he) has unit demand for a good sold by a risk-neutral
seller (she) at zero marginal cost. The buyer's initial valuation
is $L>0$. Before facing the seller, the buyer can pay a cost $c$
to raise his valuation to $H>L$. The cost $c$ is the buyer's \emph{private
type} and is drawn from distribution $F$, with positive and continuous
density $f$ on $[\underline{c},\bar{c}]$, where $\underline{c}>0$.
Let $\gamma:=H-L$ denote the investment gain. We assume $\gamma>\bar{c}$,
so investment is socially efficient for every type.\footnote{Types with $c>\gamma$ will never invest and can be incorporated without
affecting the main insights.} 

\paragraph{Profiling.}

The seller does not observe whether the buyer invests.\footnote{If the seller observed the buyer's investment decision, she would
price at the buyer's realized valuation. Thus, no type would invest,
and the unique outcome would be no investment and seller profit $L$.} Before setting the price, however, she observes a signal about the
buyer's type. A signal structure is $\{\mathbb{S},\pi\}$, where $\mathbb{S}$
is a finite set of signals and $\pi_{s}\left(c\right)$ is the probability
that type $c$ generates signal $s$. The signal distribution depends
only on the buyer\textquoteright s type, not on his investment decision.
Thus, we study profiling based on the buyer's \emph{exogenous} characteristics
rather than monitoring technologies that observe the buyer's investment
action.

\paragraph{Timing.}

First, the signal structure is announced. Second, the buyer observes
his type $c$ and chooses whether to invest. The investment cost,
if incurred, is sunk. Third, the seller observes a signal drawn from
$\pi$ and offers a take-it-or-leave-it price $p$, which the buyer
accepts whenever his payoff is nonnegative. Outside options are zero.

\paragraph{Equilibrium.}

A signal structure induces a Bayesian game. In this game, a buyer
strategy maps each type $c$ into a probability of investing. Since
the buyer's valuation is either $L$ or $H$, it is without loss to
restrict the seller's prices to $\left\{ L,H\right\} $, so a seller
strategy maps each signal $s$ to a distribution over $\left\{ L,H\right\} $.
The solution concept is Bayes-Nash equilibrium. 

A buyer-seller payoff pair $\left(u_{B},u_{S}\right)$ is \emph{attainable}
if there exists a signal structure and an equilibrium of the induced
game that yield expected buyer payoff $u_{B}$ and expected seller
payoff $u_{S}$. We henceforth drop the ``expected'' quantifier.
The buyer can guarantee payoff $0$ by not investing, and the seller
can guarantee profit $L$ by always charging $L$. Thus, payoffs are
\emph{individually rational} (IR) only if $u_{B}\ge0$ and $u_{S}\ge L$. 

\section{\protect\label{Section:Main-Result}Main Results}

\subsection{Attainable Payoffs}

We first derive two necessary bounds on attainable payoffs. Fix an
equilibrium attaining $\left(u_{B},u_{S}\right)$, and let $q_{L}\left(c\right)$
denote the equilibrium probability that type $c$ faces $p=L$. If
type $c$ invests, his payoff is $q_{L}\left(c\right)\gamma-c$; if
he does not, his payoff is $0$. Buyer optimality therefore implies
\begin{equation}
q_{L}\left(c\right)\ge c/\gamma\ \ \text{ if }c\text{ invests,}\ \ \ \ \ \ \ q_{L}\left(c\right)\le c/\gamma\ \ \text{ if }c\text{ does not invest.}\label{eq:BIC}
\end{equation}
Consider the surplus generated in equilibrium. An investing type always
trades and generates surplus $H-c$. A noninvesting type trades only
at $p=L$. By (\ref{eq:BIC}), this occurs with probability at most
$c/\gamma$, so he generates at most surplus $\left(c/\gamma\right)L$.
Since $H-c$ decreases in $c$, whereas $\left(c/\gamma\right)L$
increases in $c$, total surplus is maximized, for any given investment
probability, when the investors are the lowest-cost types. Let $\hat{\tau}$
denote the type for which $F\left(\hat{\tau}\right)$ equals the equilibrium
probability of investment. Equilibrium surplus is thus at most $SW\left(\hat{\tau}\right)$,
where 
\[
SW\left(\tau\right):=\int_{\underline{c}}^{\tau}\left(H-c\right)f\left(c\right)dc+\int_{\tau}^{\bar{c}}\frac{c}{\gamma}Lf\left(c\right)dc.
\]
Hence, the first bound is that $\left(u_{B},u_{S}\right)$ must satisfy
\begin{equation}
u_{B}+u_{S}\le SW\left(\hat{\tau}\right).\tag{SC}\label{eq:SC}
\end{equation}
Next, if the seller ignores her signal and always charges $p=H$,
she sells exactly to the types who invest. Since the measure of investing
types is $F\left(\hat{\tau}\right)$, this deviation yields payoff
$F\left(\hat{\tau}\right)H$. Thus, the second bound is 
\begin{equation}
u_{S}\ge F\left(\hat{\tau}\right)H.\tag{HC}\label{eq:HC}
\end{equation}

Our first result is that payoff attainability is fully determined
by just these two bounds. 
\begin{prop}
\label{Proposition:AC}An IR pair $\left(u_{B},u_{S}\right)$ is attainable
if and only if there exists $\tau\in\left[\underline{c},\bar{c}\right]$
such that
\begin{equation}
SW\left(\tau\right)-u_{B}\ \underbrace{\ge}_{\text{SC}}\ u_{S}\ \underbrace{\ge}_{\text{HC}}\ F\left(\tau\right)H.\tag{AC}\label{eq:AC}
\end{equation}
\end{prop}
Necessity of the \emph{attainability constraint} (\ref{eq:AC}) follows
from the preceding arguments. We establish sufficiency by construction
in Proposition \ref{Proposition:Construction} in Subsection \ref{subsection:Proof-Theorem-1}.

The two inequalities in (\ref{eq:AC}) capture the two forces of endogenous
valuation. (\ref{eq:SC}) is a \emph{surplus constraint}: there must
exist some investment level generating enough surplus to sustain the
payoffs. Since 
\[
\frac{d}{d\tau}SW\left(\tau\right)=H\left(1-\frac{\tau}{\gamma}\right)f\left(\tau\right)>0,
\]
(\ref{eq:SC}) imposes a \emph{lower} bound on investment. 

(\ref{eq:HC}) is a \emph{hold-up risk constraint}. If the mass of
investing types is too high, the seller would ignore her signal and
charge $H$; anticipating that his investment gain will be expropriated,
the buyer does not invest in the first place, reminiscent of the classic
hold-up problem. (\ref{eq:HC}) thus imposes an \emph{upper} bound
on investment and restricts surplus creation.

The first-best surplus requires investment from every type, but since
$SW\left(\bar{c}\right)<F\left(\bar{c}\right)H$, this is never attainable. 
\begin{cor}
\label{Corollary:No-FB}Any IR pair $\left(u_{B},u_{S}\right)$ that
sums to the first-best expected surplus $H-E\left[c\right]$ is unattainable. 
\end{cor}
\bigskip

Next, we use Proposition \ref{Proposition:AC} to characterize the
attainable set, beginning with its Pareto frontier. Suppose $\left(u_{B},u_{S}\right)$
is on the frontier and $\tau$ satisfies (\ref{eq:AC}) with this
pair. If (\ref{eq:SC}) were slack, raising $u_{B}$---while holding
$u_{S}$ and $\tau$ fixed---preserves (\ref{eq:AC}), contradicting
Pareto optimality. Hence, (\ref{eq:SC}) binds. If (\ref{eq:HC})
were slack, then (\ref{eq:AC}) is
\begin{equation}
SW\left(\tau\right)-u_{B}=u_{S}>F\left(\tau\right)H.\label{eq:HC-binds}
\end{equation}
Since $SW$ is strictly increasing and the second inequality in (\ref{eq:HC-binds})
is strict, replacing $\tau$ with a slightly higher $\tau'$ preserves
(\ref{eq:AC}) and makes (\ref{eq:SC}) slack, contradicting Pareto
optimality by the previous argument. Hence, (\ref{eq:HC}) must also
bind.

Since both (\ref{eq:SC}) and (\ref{eq:HC}) are binding, a pair $\left(u_{B},u_{S}\right)$
on the attainable Pareto frontier must take the form 
\begin{equation}
u_{S}=F\left(\tau\right)H\ ,\ \ \ \ u_{B}=\xi\left(\tau\right):=SW\left(\tau\right)-F\left(\tau\right)H\label{eq:Frontier-(uS,uB)}
\end{equation}
for some $\tau$.\footnote{The converse also holds: if (\ref{eq:SC}) and (\ref{eq:HC}) bind
at some $\tau$, the IR pair is on the frontier---see Lemma \ref{Lemma:Frontier-Sufficiency.}
in Appendix \ref{Section:MainAppendix}. } Furthermore, differentiating $\xi$ yields 
\begin{align}
\frac{d}{d\tau}\xi\left(\tau\right)= & -f\left(\tau\right)\tau\left(1+\frac{L}{\gamma}\right)<0,\label{eq:xi-prime}
\end{align}
so $\xi$ is strictly decreasing. Each frontier pair is therefore
determined by a \emph{unique} $\tau$, with $u_{S}$ increasing and
$u_{B}$ decreasing in $\tau$. The Pareto frontier of the attainable
set therefore runs between two points: the buyer-optimal point holds
the seller to her IR payoff $L$, and the seller-optimal point holds
the buyer to his IR payoff $0$. Denote the associated $\tau$ by
$\tau_{B}^{*}$ and $\tau_{S}^{*}$, defined by 
\begin{equation}
F\left(\tau_{B}^{*}\right)H=L\ \ \ \text{ and }\ \ \ \xi\left(\tau_{S}^{*}\right)=0,\label{eq:tau*}
\end{equation}
which satisfy $\underline{c}<\tau_{B}^{*}<\tau_{S}^{*}<\bar{c}$.\footnote{See Lemma \ref{Lemma:tau*} in Appendix \ref{Section:MainAppendix}.}
This completes the attainable Pareto frontier characterization.

Next, every IR payoff pair below the frontier is also attainable.
To see why, suppose $\left(u_{B},u_{S}\right)$ is on the frontier.
From the previous discussion, there exists $\tau\in\left[\tau_{B}^{*},\tau_{S}^{*}\right]$
such that 
\begin{equation}
SW\left(\tau\right)-u_{B}\ =\ u_{S}\ =\ F\left(\tau\right)H.\label{eq:AC-binding}
\end{equation}
If $u_{B}$ in (\ref{eq:AC-binding}) is replaced by any $u_{B}'<u_{B}$,
(\ref{eq:AC}) still holds at the same $\tau$, so $\left(u_{B}',u_{S}\right)$
is attainable by Proposition \ref{Proposition:AC}. If $u_{S}$ in
(\ref{eq:AC-binding}) is replaced by some $u_{S}'<u_{S}$, there
exists $\tau'<\tau$ that satisfies (\ref{eq:AC}) with $u_{B}$---this
is because $\xi'<0$, so $F\left(\tau\right)H$ increases faster in
$\tau$ than $SW\left(\tau\right)$---so $\left(u_{B},u_{S}'\right)$
is attainable by Proposition \ref{Proposition:AC}.

Let $\bar{U}_{S}\left(u_{B}\right):=F\left(\xi^{-1}\left(u_{B}\right)\right)H$
denote the seller's payoff on the frontier at buyer payoff $u_{B}$. 
\begin{thm}
\label{Theorem:Characterization}The pair $\left(u_{B},u_{S}\right)$
is attainable if and only if
\[
0\le u_{B}\le\xi\left(\tau_{B}^{*}\right)\ \ \ \text{ and }\ \ \ L\le u_{S}\le\bar{U}_{S}\left(u_{B}\right).
\]
\end{thm}
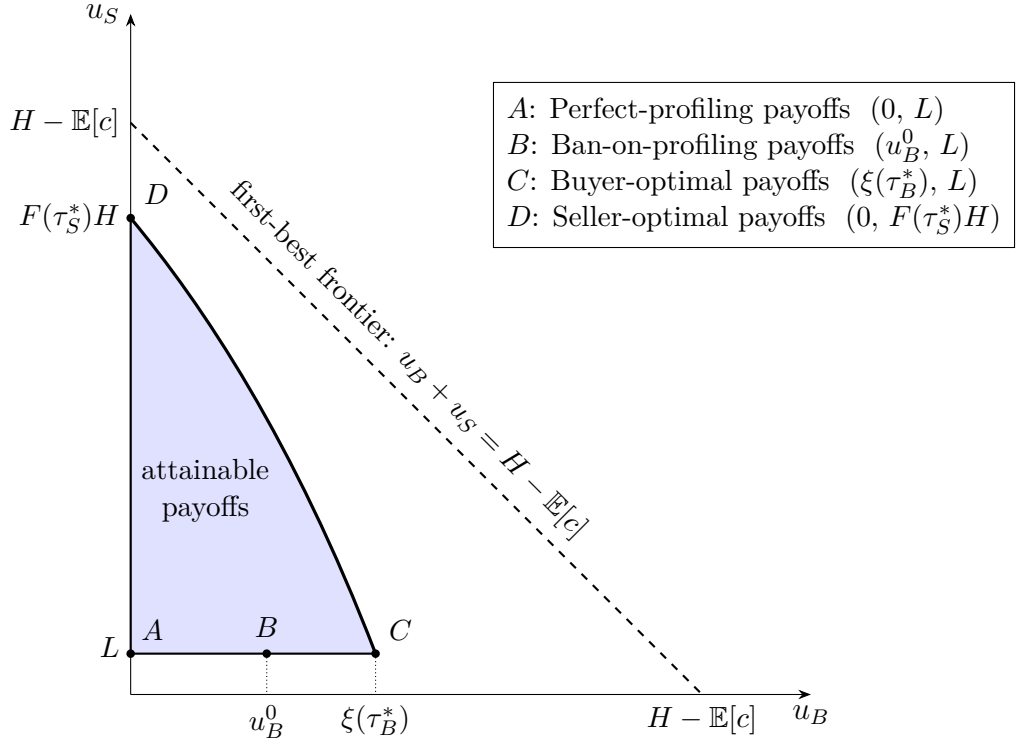
\begin{figure}[h]
\caption{\protect\label{Figure:Attainable-Payoffs}Attainable Payoffs}

\centering{}

\begin{tikzpicture}[x=1.8cm,y=1.8cm,>=Stealth]

\coordinate (A) at (0,0.3);      
\coordinate (B) at (1,0.3);   
\coordinate (C) at (1.8,0.3);    
\coordinate (D) at (0,3.5);      


\fill[blue!12]
  (D) .. controls (0.533,2.860) and (1.20,1.860) .. (C)
  -- (A) -- cycle;

\draw[dashed, thick] (0,4.2) -- (4.2,0)
  node[pos=0.45, sloped, above=1pt, font=\small]
  {first-best frontier: $u_B+u_S=H-\mathbb{E}[c]$};

\draw[very thick]
  (D) .. controls (0.533,2.860) and (1.200,1.860) .. (C);

\draw[thick] (A) -- (C);
\draw[thick] (A) -- (D);

\draw[->] (0,0) -- (5,0) node[below] {$u_B$};
\draw[->] (0,0) -- (0,5) node[left] {$u_S$};

\node[below left, font=\small] at (0,0.5) {$L$};
\node[left, font=\small] at (0,3.5) {$F(\tau_S^*)H$};
\node[left, font=\small] at (0,4.2) {$H-\mathbb{E}[c]$};
\node[below, font=\small] at (4.2,0) {$H-\mathbb{E}[c]$};

\draw[densely dotted] (B) -- (1,0) node[below, font=\small] {$u_B^0$};
\draw[densely dotted] (C) -- (1.8,0)  node[below, font=\small] {$\xi(\tau_B^*)$};

\node[font=\small, align=center] at (0.55,1.5) {attainable\\ payoffs};

\fill (A) circle (1.6pt) node[above right=1pt and 0pt, font=\small] {$A$};
\fill (B) circle (1.6pt) node[above=2pt, font=\small] {$B$};
\fill (C) circle (1.6pt) node[above right=1pt and 1pt, font=\small] {$C$};
\fill (D) circle (1.6pt) node[above right=1pt and 1pt, font=\small] {$D$};

\node[anchor=north east, align=left, font=\small,
      draw, thin, inner sep=5pt] at (6.5,4.5)
  {$A$: Perfect-profiling payoffs $ \   (0,\,L)$\\
   $B$: Ban-on-profiling payoffs $\ (u_B^0,\,L)$\\
   $C$: Buyer-optimal payoffs $ \   (\xi(\tau_B^*),\,L)$\\
   $D$: Seller-optimal payoffs $\   (0,\,F(\tau_S^*)H)$};

\end{tikzpicture}
\end{figure}

Figure \ref{Figure:Attainable-Payoffs} illustrates the attainable
set. Its Pareto frontier is strictly below the first-best frontier,
as established in Corollary \ref{Corollary:No-FB}. The shape of the
Pareto frontier further shows that (\ref{eq:HC}) restricts not only
surplus creation but also \emph{surplus division}, as we make precise
next. 
\begin{prop}
\label{Proposition:Pareto-frontier}The Pareto frontier of the attainable
payoffs is strictly concave, with $\bar{U}_{S}'\left(u_{B}\right)<-1$.
Consequently, as buyer payoff increases along the frontier, total
welfare strictly falls, and at an increasing rate.
\end{prop}
The Pareto frontier of the BBM ``surplus triangle'' has slope $-1$,
so any division of the maximum surplus is attainable. In contrast,
surplus division on the attainable Pareto frontier here is restricted---increasing
buyer payoff requires decreasing total welfare. The reason is that
(\ref{eq:HC}) causes the seller's share of the surplus on the frontier
to rise disproportionately with investment. To see this, from (\ref{eq:Frontier-(uS,uB)}),
the seller's payoff is $F\left(\tau\right)H$, the payoff she could
get from ``holding up the buyer'' by always pricing $H$. Marginally
increasing the investment cutoff implies 
\begin{equation}
\frac{d}{d\tau}F\left(\tau\right)H=Hf\left(\tau\right)\ >\ H\left(1-\frac{\tau}{\gamma}\right)f\left(\tau\right)=\frac{d}{d\tau}SW\left(\tau\right).\label{eq:Surplus-Tradeoff}
\end{equation}
(\ref{eq:Surplus-Tradeoff}) shows that when investment increases,
the surplus that the seller captures, $F\left(\tau\right)H$, always
grows \emph{faster} than the surplus that the investment creates,
$SW\left(\tau\right)$. Increasing the buyer's payoff therefore requires
\emph{decreasing} investment and, with it, the surplus generated.
A stark implication that follows is that the highest attainable total
welfare is obtained \emph{only} at the seller-optimal payoff pair
$\left(0,F\left(\tau_{S}^{*}\right)H\right)$, which holds the buyer
to his outside option value.

\bigskip

Next, we discuss the value of two common profiling policies. 

\paragraph{Perfect profiling.}

Suppose the seller perfectly observes $c$. The unique equilibrium
payoff is $\left(u_{B},u_{S}\right)=(0,L)$. If the seller expects
type $c$ to invest, she charges $H$, which deters investment; if
she expects him not to invest, she charges $L$, which induces investment.
Thus, both players randomize in equilibrium, leading to payoff pair
$\left(0,L\right)$.\footnote{In equilibrium, every type invests with probability $L/H$, and the
seller sets $p=L$ with probability $c/\gamma$ for type $c$. Since
every type plays ``not invest'' with positive probability, his equilibrium
payoff must equal that from not investing, which is $0$. Since the
seller charges $L$ with positive probability against every type,
her equilibrium payoff equals that from charging $L$, which is $L$.} 

\paragraph{Ban on profiling.}

Suppose that the seller observes no signal about $c$. By the same
argument as before, she must randomize in equilibrium, so her payoff
is $L$. Let $q_{L}^{0}$ denote the equilibrium probability of $p=L$.
By (\ref{eq:BIC}), type $c$ invests if and only if $c\le q_{L}^{0}\gamma=:\tau^{0}$,
and the seller's indifference between the two prices requires $\tau^{0}$
to satisfy $F\left(\tau^{0}\right)H=L$, so $\tau^{0}=F^{-1}\left(L/H\right)=\tau_{B}^{*}$
from (\ref{eq:tau*}). The buyer's payoff is therefore
\begin{equation}
u_{B}^{0}=\int_{\underline{c}}^{\tau_{B}^{*}}\left(\tau_{B}^{*}-c\right)f\left(c\right)dc,\ \ \ \label{eq:u^0_B}
\end{equation}
with $0<u_{B}^{0}<\xi\left(\tau_{B}^{*}\right)$.\footnote{See Lemma \ref{Lemma:u^0_B} in Appendix \ref{Section:MainAppendix}.}

\bigskip

The payoffs from these two policies relative to the attainable payoff
set are illustrated in Figure \ref{Figure:Attainable-Payoffs}. Their
positions have implications for the regulation of consumer data. First,
more information need \emph{not} benefit the seller and may even decrease
profit. This contrasts with fixed-valuation environments, where information
can never hurt the seller. Firms may therefore voluntarily limit their
data collection when valuations are endogenous, but such restraint
does not necessarily improve buyer welfare---the seller-optimal outcome
likewise gives the buyer zero payoff.

Second, banning profiling offers some buyer protection, since $u_{B}^{0}>0$,
but it is not buyer-optimal. Since $u_{B}^{0}<\xi\left(\tau_{B}^{*}\right)$,
the buyer can do better with the seller having some information. In
fact, since the seller's payoff is $L$, some information can make
\emph{both} parties better off. While a similar conclusion also holds
in fixed-valuation environments (such as BBM), the mechanism here
is different. There, information benefits buyers by allowing the seller
to open up low-price markets to buyers who would otherwise be excluded.
Here, information benefits the buyer by mitigating hold-up risk and
allowing him to retain part of the surplus created by investment.
The relevant policy question is therefore not whether the seller should
have more or less information, but which buyer types the information
concerns. The next subsection makes this point precise.

\bigskip

\subsection{\protect\label{subsection:Proof-Theorem-1}Attaining the Payoffs}

We now construct, for each payoff pair satisfying (\ref{eq:AC}),
a signal structure and equilibrium that attain it. It suffices to
use signals that serve as price recommendations, so let $\mathbb{S}=\left\{ L,H\right\} $.
For $\phi\in\left[\underline{c},\bar{c}\right]$, define the signal
structure $\pi^{\phi}$ by 
\begin{equation}
\pi_{L}^{\phi}\left(c\right)=\begin{cases}
1 & \text{ if }c\le\phi\\
c/\gamma & \text{ if }c>\phi
\end{cases}\ \ \ \ ;\ \ \ \pi_{H}^{\phi}\left(c\right)=1-\pi_{L}^{\phi}\left(c\right)\ \text{ for all }c\in\left[\underline{c},\bar{c}\right].\label{eq:pi-phi}
\end{equation}

\begin{prop}
\label{Proposition:Construction}Suppose $\left(u_{B},u_{S}\right)$
is IR and satisfies (\ref{eq:AC}). Then there exist $\phi$ and $\tau$,
with $\underline{c}\le\phi<\tau<\bar{c}$, such that
\begin{equation}
u_{B}=\int_{\underline{c}}^{\phi}\left(\gamma-c\right)f\left(c\right)dc\ \ \ \text{ and }\ \ \ SW\left(\tau\right)=u_{B}+u_{S}.\label{eq:AttainRule}
\end{equation}
Under signal structure $\pi^{\phi}$, the following profile is an
equilibrium attaining $\left(u_{B},u_{S}\right)$: types $c\le\tau$
invest, types $c>\tau$ do not, and the seller sets $p=s$.
\end{prop}
\begin{figure}[h]
\centering \ \caption{\protect\label{Figure:kappa-neutral}Dual-cutoff Equilibrium}
\bigskip{}
\begin{tikzpicture} 

\draw [line width=0.3mm, dotted] (0,-2.5)--(0,1.5);
\draw [line width=0.3mm, dotted] (15,-2.5)--(15,1.5);

\draw [line width=0.2mm, -] (0,-0.25)--(0,0.25);
\node [below left] at (0,0) {$\underline{c}$};
\draw [line width=0.6mm, dashed, blue] (0,0)--(5,0);

\draw [line width=0.2mm, -] (5,-0.25)--(5,0.25);
\node [below right] at (5,0) {$\phi$};
\draw [line width=0.6mm, dotted, green!75!black] (5,0)--(9,0);
\draw [line width=0.2mm, -] (9,-0.25)--(9,0.25);
\node [below left] at (9,0) {$\tau$};

\draw [line width=0.6mm, -, red] (9,0)--(15,0);
\draw [line width=0.3mm, -] (15,-0.25)--(15,0.25);
\node [below right] at (15,0) {$\bar{c}$};

\draw [line width=0.3mm, dotted] (5,0)--(5,-2.5);
\node[text width=5cm, align=center] at (2.5,-1.7) {$\pi^{\phi}_L(c)=1$ \\ \smallskip Buyer payoff  = $\gamma-c$};
\node[text width=5cm, align=center] at (9.6,-1.7) {$\pi^{\phi}_L(c) = \frac{c}{\gamma}$ \\\smallskip Buyer payoff = 0};

\draw [line width=0.3mm, dotted] (9,0)--(9,1.5);
\node[text width=6cm, align=center] at (4.5,1) {Types who invest};
\node[text width=6cm, align=center] at (12,1) {Types who do not invest};

    \end{tikzpicture}
\end{figure}
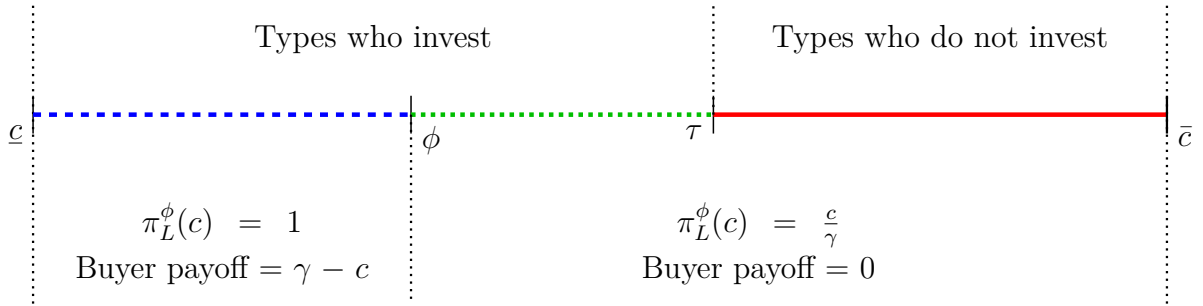

Figure \ref{Figure:kappa-neutral} illustrates the equilibrium constructed
in Proposition \ref{Proposition:Construction}. By (\ref{eq:BIC}),
when the seller sets $p=s$, every type $c>\phi$ is indifferent between
investing and not investing, whereas types $c\le\phi$ strictly prefer
to invest. Thus, the buyer strategy of investing if and only if $c\le\tau$
is supported whenever $\phi<\tau$. Each type $c\le\phi$ earns payoff
$\gamma-c$; each type $c>\phi$ is indifferent and thus earns the
payoff from not investing, which is zero. We verify seller obedience
to setting $p=s$ given this buyer strategy below.

The equilibrium separates the two key margins identified in the previous
subsection into two cutoffs: $\tau$ determines surplus creation,
and $\phi$ determines surplus division. On $\left[\underline{c},\phi\right]$,
the buyer invests and retains the entire investment gain. On $(\phi,\tau]$,
the buyer still invests, but it is the seller who, in expectation,
captures the investment gain, leaving the buyer with zero payoff.
On $(\tau,\bar{c}]$, the buyer does not invest, and trade breaks
down when the seller sets price $H$, which occurs with probability
$1-c/\gamma$. This inefficiency is required to incentivize types
$c>\tau$ not to invest and maintain valuation $L$, which in turn
is necessary to incentivize the seller to charge price $L$ with positive
probability in equilibrium, allowing types $c\le\phi$ to then retain
their investment surplus. 

\subsubsection*{Proof of Proposition \ref{Proposition:Construction}.}
\begin{proof}
Suppose $\hat{\tau}$ satisfies (\ref{eq:AC}) for $\left(u_{B},u_{S}\right)$.
Since $SW$ is strictly increasing with $SW\left(\underline{c}\right)<L\le u_{B}+u_{S}\le SW\left(\hat{\tau}\right)$,
there exists $\tau\le\hat{\tau}$ such that 
\begin{equation}
SW\left(\tau\right)-u_{B}=u_{S}\ge F\left(\tau\right)H,\label{eq:tau-prime}
\end{equation}
where the inequality follows from $F\left(\tau\right)H\le F\left(\hat{\tau}\right)H\le u_{S}$.
Furthermore, writing $H-c=\left(\gamma-c\right)+L$, we obtain
\[
SW\left(\tau\right)=\int_{\underline{c}}^{\tau}\left(\gamma-c\right)f\left(c\right)dc+L\left[F\left(\tau\right)+\int_{\tau}^{\bar{c}}\frac{c}{\gamma}f\left(c\right)dc\right]<\int_{\underline{c}}^{\tau}\left(\gamma-c\right)f\left(c\right)dc+L,
\]
where the inequality is because the square bracketed term is the total
probability of trade, which is less than one. Hence 
\[
u_{B}=SW\left(\tau\right)-u_{S}\le SW\left(\tau\right)-L<\int_{\underline{c}}^{\tau}\left(\gamma-c\right)f\left(c\right)dc,
\]
so there exists $\phi\in[\underline{c},\tau)$ such that $u_{B}=\int_{\underline{c}}^{\phi}\left(\gamma-c\right)f\left(c\right)dc$.
This proves existence of $\phi$ and $\tau$ satisfying (\ref{eq:AttainRule}).

\bigskip

\textbf{Buyer optimality and payoff.} Already shown above.

\bigskip

\textbf{Seller optimality and payoff.} Given the buyer's strategy
and signal structure $\pi^{\phi}$, let $Q_{s}$ denote the probability
of signal $s$, and let $Q_{s}^{I}$ denote the joint probability
of receiving signal $s$ \emph{and }the buyer having invested.\footnote{$Q_{s}=\int_{\underline{c}}^{\bar{c}}\pi_{s}^{\phi}\left(c\right)f\left(c\right)dc$;
$Q_{s}^{I}=\int_{\underline{c}}^{\tau}\pi_{s}^{\phi}\left(c\right)f\left(c\right)dc$.} The seller's payoff is $u_{S}$ by construction from (\ref{eq:AttainRule})
and can be expressed as\footnote{If the seller sets $p=L$, trade always occurs; if she sets $p=H$,
trade occurs only if the buyer invested.} 
\[
u_{S}=Q_{L}L+Q_{H}^{I}H.
\]

The seller obeys $L$ if $L\ge\left(Q_{L}^{I}/Q_{L}\right)H$. Using
$Q_{L}^{I}+Q_{H}^{I}=F\left(\tau\right)$,
\[
Q_{L}L-Q_{L}^{I}H=Q_{L}L+Q_{H}^{I}H-F\left(\tau\right)H=u_{S}-F\left(\tau\right)H\ge0,
\]
where the last inequality follows from (\ref{eq:tau-prime}). Thus,
the seller obeys $L$. 

Next, the seller obeys $H$ if $\left(Q_{H}^{I}/Q_{H}\right)H\ge L$.
Using $Q_{L}+Q_{H}=1$,
\[
Q_{H}^{I}H-Q_{H}L=Q_{H}^{I}H+Q_{L}L-L=u_{S}-L\ge0,
\]
so the seller also obeys $H$.
\end{proof}
\bigskip

The dual-cutoff equilibrium also shows which information hurts and
benefits the buyer. The seller only cares about the buyer's valuation,
so profiling can be considered accurate when investing types $c\le\tau$
generate signal $H$ and noninvesting types $c>\tau$ generate signal
$L$. Accuracy harms the investing types, since they retain their
investment surplus only when mis-profiled and offered $L$. In contrast,
noninvesting types' payoff is zero regardless. By setting $\pi_{L}\left(c\right)=c/\gamma$
for $c>\tau$, the noninvesting types' profiling accuracy is maximized
subject to maintaining their incentive to not invest.\footnote{This property is crucial for attaining payoffs on the Pareto frontier.
Online appendix \ref{OA:ParetoNSC} examines this point further.} This accuracy in turn lets the seller tolerate mis-profiling of the
investing types---allowing type $c\le\phi$ to be mis-profiled with
probability one---who strictly benefit from it. This also explains
why simple policies like a ban on profiling are inefficient: a ban
simply spreads error across all types and wastes some of it on noninvesting
types, for whom it has no welfare value.

\section{\protect\label{Section:MPC}Comparative Statics of the Type Distribution}

How do the attainable payoffs depend on the distribution of the buyer's
type? With fixed valuations, this question is less meaningful since
any division of the first-best surplus is attainable. By contrast,
with endogenous valuation, the type distribution also determines the
severity of hold-up risk, which then affects the attainable set. In
this section, we consider the effects of a mean-preserving spread
of the type distribution, which leaves the first-best surplus $H-\mathbb{E}[c]$
unchanged.

Let $G$ be another atomless distribution with full support on $[\underline{c},\bar{c}]$.
For each $J\in\{F,G\}$, let $u_{B}^{J*}$ denote the highest attainable
buyer payoff and $\overline{U}_{S}^{J}(\cdot)$ the corresponding
Pareto frontier of the attainable payoffs. Since the IR payoffs are
the same under $G$ and $F$, the comparison across distributions
is governed by the frontier.
\begin{prop}
\label{prop:mps} If $F$ is a mean-preserving spread of $G$, then
\begin{equation}
u_{B}^{F*}\geq u_{B}^{G*}\quad\text{and}\quad\overline{U}_{S}^{F}(u_{B})\geq\overline{U}_{S}^{G}(u_{B})\quad\text{for every }u_{B}\in[0,u_{B}^{G*}].\label{eq:MPS-outcome}
\end{equation}
Consequently, the attainable payoff set under $F$ contains that under
$G$.\footnote{$F$ is a mean-preserving spread of $G$ if $\int_{\underline{c}}^{c}F\left(z\right)dz\ge\int_{\underline{c}}^{c}G\left(z\right)dz$
for all $c\in\left[\underline{c},\bar{c}\right]$, with equality at
$c=\bar{c}$. If the inequality is strict for all $c\in\left(\underline{c},\bar{c}\right)$,
then both inequalities in (\ref{eq:MPS-outcome}) are also strict.}
\end{prop}
To provide some intuition, substituting $H=\gamma+L$ into $\xi(\tau)=SW(\tau)-F(\tau)H$
and simplifying yields 
\begin{equation}
\xi(\tau)=\frac{L}{\gamma}\,\mathbb{E}[c]\;-\;\frac{H}{\gamma}\int_{\underline{c}}^{\tau}cf(c)\,dc.\label{eq:xi-identity}
\end{equation}
Fixing the mean of $c$, the two players' frontier payoffs are determined
by different features of the investing set. The seller's payoff $F(\tau)H$
depends only on the \emph{mass} of the investing types, whereas the
surplus left for the buyer depends on the aggregate\emph{ }cost\emph{
}of the investing types, $\int_{\underline{c}}^{\tau}cf(c)\,dc$.
A more dispersed distribution allows any given mass of investment
to be supplied by lower-cost types, thereby relaxing the trade-off
between surplus creation and division identified in Proposition \ref{Proposition:Pareto-frontier}.

Proposition \ref{prop:mps} has implications for the buyer's own information.
Suppose the buyer is uncertain about his own investment cost and instead
observes an imperfect signal about $c$. For example, a firm adapting
its operations around a vendor may only estimate the disruption cost
in advance, with the true cost realized during implementation. Let
$x$ denote his posterior expected cost from the signal. Because $c$
enters his payoff linearly, his investment decision depends only on
$x$. If profiling is based on $x$, the resulting game is equivalent
to the baseline model with the distribution of $x$ as the type distribution.
A Blackwell-more-informative learning technology induces a mean-preserving
spread of posterior expected costs. Proposition \ref{prop:mps} therefore
implies that a better-informed buyer supports a larger set of attainable
payoffs, including outcomes that raise both players' payoffs. 

\section{\protect\label{Section:Multiple-Investments}Multiple Investment
Options}

The binary investment assumption is not essential. What matters is
that the investment option with the highest average return is also
the efficient option for all types. When this holds, the set of attainable
payoffs is unchanged because we can then restrict attention to the
buyer choosing only between not investing and the efficient option,
which collapses to the binary case.

To illustrate this, suppose the buyer chooses an investment option
$i\in\left\{ 0,1,\dots,n\right\} $. Option $0$ denotes no investment
and yields $v_{0}=L$. Each option $i\ge1$ raises valuation to $v_{i}$
at cost $ic$, where $v_{i}$ is increasing in $i$ and $v_{1}>L$.
The baseline model is $n=1$ and $v_{1}=H$.
\begin{assumption}
\label{Assumption}$\left(v_{n}-L\right)/n\ge\left(v_{i}-L\right)/i$
for all $i=1,\dots,n$, and $v_{n}-n\bar{c}>L$.
\end{assumption}
The first condition says that option $n$ has the highest average
return, and the second condition says that choosing option $n$ is
socially efficient relative to no investment even for the highest-cost
type.\footnote{To simplify exposition, we assume that the efficient option is the
highest option $n$. This is not necessary. Online appendix \ref{OA:Multiple-Investments}
shows that the same characterization holds when the efficient option
is an interior option.} These two conditions jointly imply that option $n$ is the most efficient
option for \emph{every} type.
\begin{lem}
\label{Lemma:MultiInvestment}Under Assumption \ref{Assumption},
\begin{equation}
v_{n}-nc>v_{i}-ic\ \ \ \text{for all }i<n\text{ and }c\in\left[\underline{c},\bar{c}\right].\label{eq:Efficiency-MultipleAction}
\end{equation}
\end{lem}
Define

\[
SW_{n}(\tau):=\int_{\underline{c}}^{\tau}(v_{n}-nc)f(c)\,dc+\int_{\tau}^{\bar{c}}\frac{nc}{v_{n}-L}Lf(c)\,dc.
\]

\begin{prop}
\label{Proposition:Multiple-Investment}Under Assumption \ref{Assumption},
an IR pair $(u_{B},u_{S})$ is attainable if and only if there exists
$\tau\in[\underline{c},\bar{c}]$ such that
\begin{equation}
SW_{n}(\tau)-u_{B}\ \underbrace{\ge}_{\text{SC'}}\ u_{S}\ \underbrace{\ge}_{\text{HC'}}\ F(\tau)v_{n}\tag{AC'}\label{eq:AC-prime}
\end{equation}
 
\end{prop}
With Proposition \ref{Proposition:Multiple-Investment}, the payoff
characterization from the binary-investment model carries over by
replacing $H$ with $v_{n}$ and the investment cost $c$ with $nc$.
The additional intermediate options therefore do not affect the attainable
set.

\subsubsection*{Proof of Proposition \ref{Proposition:Multiple-Investment}.}
\begin{proof}
For necessity of (\ref{eq:AC-prime}), fix an equilibrium attaining
$\left(u_{B},u_{S}\right)$, and let $\tau$ be defined such that
$F\left(\tau\right)$ equals the probability that the buyer chooses
$n$. The seller cannot strictly benefit from deviating to always
setting $p=v_{n}$, so $u_{S}\ge F(\tau)v_{n}$, which establishes
(HC'). 

For (SC'), first note that the surplus generated from the buyer choosing
option $n$ is $v_{n}-nc$. Now consider a type $c$ who chooses option
$j<n$ with positive probability, and let $q_{j}(c)$ denote the probability
that he is offered $p\le v_{j}$. Deviating from option $j$ to $n$
would raise his payoff by at least $q_{j}(c)(v_{n}-v_{j})-c(n-j)$.
Buyer optimality therefore requires
\begin{equation}
q_{j}(c)\le c\frac{n-j}{v_{n}-v_{j}}.\label{eq:q_j(c)}
\end{equation}
Assumption \ref{Assumption} implies that 
\begin{equation}
v_{j}\le L+\frac{j}{n}(v_{n}-L)=\frac{n-j}{n}L+\frac{j}{n}v_{n}\ \ \ \ \forall j\le n.\label{eq:v_j}
\end{equation}
Since trade with a buyer who has chosen option $j$ occurs if and
only if $p\le v_{j}$, the expected surplus generated when type $c$
chooses option $j$ is at most
\begin{align}
q_{j}(c)v_{j}-jc\  & \le\ \left[\frac{n-j}{v_{n}-v_{j}}v_{j}-j\right]c\ \le\ \frac{nL}{v_{n}-L}c.\label{eq:surplus-inequality}
\end{align}
The first inequality in (\ref{eq:surplus-inequality}) follows from
substituting in the upper bound of $q_{j}\left(c\right)$ from (\ref{eq:q_j(c)}).
The square-bracketed term is strictly increasing in $v_{j}$. Substituting
in the upper bound of $v_{j}$ from (\ref{eq:v_j}) and simplifying
gives the second inequality in (\ref{eq:surplus-inequality}). Thus,
the surplus generated when type $c$ chooses option $j\ne n$ is at
most $\frac{nc}{v_{n}-L}L$. The argument for (\ref{eq:SC}) then
yields (SC') here.

Sufficiency of (\ref{eq:AC-prime}) uses the same dual-cutoff construction
in Proposition \ref{Proposition:Construction}. The signal structure
is supported on signals $\mathbb{S}=\left\{ L,v_{n}\right\} $, with
distribution
\[
\pi_{L}(c)=1\text{ for }c\le\phi,\qquad\pi_{L}(c)=\frac{nc}{v_{n}-L}\text{ for }c>\phi.
\]
The buyer strategy is: type $c\le\tau$ chooses option $n$, and type
$c>\tau$ does not invest. Seller obedience is exactly as in Proposition
\ref{Proposition:Construction}. For buyer optimality, we must additionally
check that no type has a profitable deviation to an intermediate option
$j\in\left\{ 1,\dots,n-1\right\} $. Type $c$'s payoff from option
$j$ is $w_{j}\left(c\right):=\pi_{L}\left(c\right)\left[v_{j}-L\right]-jc$,
while option $n$ yields $w_{n}\left(c\right):=\pi_{L}\left(c\right)\left[v_{n}-L\right]-nc$.
For $c\le\phi$, $\pi_{L}\left(c\right)=1$, so $w_{n}\left(c\right)>w_{j}\left(c\right)$
by (\ref{eq:Efficiency-MultipleAction}). For $c>\phi$, 
\[
w_{j}\left(c\right)=\frac{nc}{v_{n}-L}(v_{j}-L)-cj=c\left[\frac{n(v_{j}-L)}{v_{n}-L}-j\right]\le0,
\]
where the inequality is the first condition of Assumption \ref{Assumption},
so option $j$ is weakly worse than not investing. This establishes
buyer optimality.
\end{proof}

\section{\protect\label{Section:Conclusion}Concluding Remarks}

We have studied the welfare implications of providing the seller with
buyer information when the buyer's valuation is endogenous to pre-trade
investment. Payoff attainability in such environments is determined
by a hold-up risk constraint that restricts not only how much surplus
can be created, but also how that surplus can be divided. In contrast
to fixed-valuation environments, the first-best surplus is unattainable,
information may hurt the seller, and protecting buyer welfare can
require reducing investment, even though investment is socially efficient.
Beyond characterizing the payoff outcomes, we show how each is attained
within a class of simple dual-cutoff equilibria and identify which
information hurts the buyer and which is payoff-neutral. 

More broadly, endogenous valuation raises informational questions
beyond those studied here. We have focused on private information
about the buyer's investment cost, but the buyer could instead hold
private information about the value his investment creates. Relatedly,
the seller could receive information directly about the buyer's investment
choice, rather than only about the characteristics that influence
it.\footnote{\citet{Lau_2008(RJE)} and \citet{Nguyen_Tan_2019(RJE)} study this
information design problem but do not allow the buyer to have ex ante
private information.} We leave these questions to future research.

\appendix

\section{\protect\label{Section:MainAppendix}Appendix: Supporting Results}
\begin{lem}
\label{Lemma:Frontier-Sufficiency.}An IR pair $\left(u_{B},u_{S}\right)$
is on the Pareto frontier if and only if there exists $\tau$ such
that $SW\left(\tau\right)-u_{B}=u_{S}=F\left(\tau\right)H$.
\end{lem}
\begin{proof}
Necessity is established in the main text. For sufficiency, suppose
there exists $\tau$ such that $SW(\tau)-u_{B}=u_{S}=F(\tau)H$. Suppose,
toward a contradiction, that $(u_{B},u_{S})$ is not on the Pareto
frontier. Then there exists an attainable pair $(u'_{B},u'_{S})$
with $u'_{B}\geq u_{B}$, $u'_{S}\geq u_{S}$, and at least one inequality
strict. By Proposition \ref{Proposition:AC}, there exists $\tau'$
satisfying (\ref{eq:AC}) for $(u'_{B},u'_{S})$. Hence
\[
\xi(\tau')=SW(\tau')-F(\tau')H\geq u'_{B}\geq u_{B}=\xi(\tau).
\]
Since $\xi$ is strictly decreasing, $\tau'\leq\tau$. Since $SW$
is strictly increasing, $SW(\tau')\leq SW(\tau)=u_{B}+u_{S}$. But
$(u'_{B},u'_{S})$ strictly Pareto-dominates $(u_{B},u_{S})$, so
$u'_{B}+u'_{S}>u_{B}+u_{S}=SW(\tau)$, which contradicts the surplus
constraint $u'_{B}+u'_{S}\leq SW(\tau')$. 
\end{proof}
\begin{lem}
\label{Lemma:tau*}$\tau_{B}^{*}$ and $\tau_{S}^{*}$---defined
in (\ref{eq:tau*})---satisfy $\underline{c}<\tau_{B}^{*}<\tau_{S}^{*}<\bar{c}$.
If $\left(u_{B},u_{S}\right)$ is on the Pareto frontier, there exists
a unique $\tau\in\left[\tau_{B}^{*},\tau_{S}^{*}\right]$ such that
$SW\left(\tau\right)-u_{B}=u{}_{S}=F\left(\tau\right)H$.
\end{lem}
\begin{proof}
Since $0<L/H<1$ and $F$ has full support, $\tau_{B}^{*}=F^{-1}(L/H)$
exists, with $\underline{c}<\tau_{B}^{*}<\bar{c}$. Next, $\xi(\underline{c})=SW(\underline{c})>0$
and $\quad\xi(\bar{c})=SW(\bar{c})-H=-\mathbb{E}[c]<0$. Since $\xi$
is continuous and strictly decreasing, there exists a unique $\tau_{S}^{*}\in(\underline{c},\bar{c})$
such that $\xi(\tau_{S}^{*})=0$. It remains to show $\tau_{B}^{*}<\tau_{S}^{*}$.
Using $F(\tau_{B}^{*})=L/H$,
\begin{equation}
\xi(\tau_{B}^{*})=\frac{H}{\gamma}\left[\frac{L}{H}\mathbb{E}[c]-\int_{\underline{c}}^{\tau_{B}^{*}}cf(c)\,dc\right].\label{eq:xi(tau*_B)}
\end{equation}
The bracketed term is strictly positive because
\begin{equation}
\int_{\underline{c}}^{\tau_{B}^{*}}cf(c)\,dc=F(\tau_{B}^{*})\mathbb{E}[c\mid c\leq\tau_{B}^{*}]<F(\tau_{B}^{*})\mathbb{E}[c]=\frac{L}{H}\mathbb{E}[c].\label{eq:inequality}
\end{equation}
Therefore $\xi(\tau_{B}^{*})>0$, so $\tau_{B}^{*}<\tau_{S}^{*}$.
Finally, if $(u_{B},u_{S})$ is on the Pareto frontier, Lemma \ref{Lemma:Frontier-Sufficiency.}
implies that there exists $\tau$ such that $u_{S}=F(\tau)H$ and
$u_{B}=\xi(\tau)$. IR requires $u_{S}\geq L$ and $u_{B}\geq0$.
Hence $\tau\geq\tau_{B}^{*}$ and $\tau\leq\tau_{S}^{*}$. Uniqueness
follows from the strict monotonicity of $\xi$.
\end{proof}
\begin{lem}
\label{Lemma:u^0_B}$u_{B}^{0}$---defined in (\ref{eq:u^0_B})---is
strictly lower than $\xi\left(\tau_{B}^{*}\right)$.
\end{lem}
\begin{proof}
Using $F\left(\tau_{B}^{*}\right)=L/H$ and (\ref{eq:xi(tau*_B)}),
\[
\xi(\tau_{B}^{*})-u_{B}^{0}=\frac{L}{\gamma}\int_{\tau_{B}^{*}}^{\bar{c}}(c-\tau_{B}^{*})f(c)\,dc>0.
\]
\end{proof}

\section{\protect\label{Section:Proofs}Appendix: Omitted Proofs}

\subsection*{Proof of Proposition \ref{Proposition:Pareto-frontier}}
\begin{proof}
From Lemma \ref{Lemma:tau*}, along the Pareto frontier, $u_{S}=F\left(\tau\right)H$
and $u_{B}=\xi\left(\tau\right)$, so 
\[
\bar{U}_{S}'\left(u_{B}\right)=\frac{du_{S}/d\tau}{du_{B}/d\tau}=\frac{f\left(\tau\right)H}{\xi'\left(\tau\right)}=-\frac{H}{\tau\left(1+L/\gamma\right)}=-\frac{\gamma}{\tau},
\]
where the last step uses $\gamma+L=H$. Since any frontier cutoff
satisfies $\tau\le\tau_{S}^{*}<\bar{c}<\gamma$, $-\gamma/\tau<-1$.
$\tau=\xi^{-1}\left(u_{B}\right)$, so $\tau$ falls with $u_{B}$;
thus, $\bar{U}_{S}$ is strictly concave. Since total welfare on the
frontier is $W\left(u_{B}\right):=u_{B}+\bar{U}_{S}\left(u_{B}\right)$,
we have $W'\left(u_{B}\right)=1+\bar{U}_{S}'\left(u_{B}\right)<0$.
\end{proof}

\subsection*{Proof of Proposition \ref{prop:mps}}
\begin{proof}
For $J\in\{F,G\}$, define the Lorenz curve $\mathcal{L}_{J}(q):=\frac{1}{\mathbb{E}[c]}\int_{\underline{c}}^{J^{-1}(q)}c\,dJ(c)$
for $q\in[0,1]$, which is strictly increasing. Since $F$ is a mean-preserving
spread of $G$, $\mathcal{L}_{F}(q)\le\mathcal{L}_{G}(q)$ for all
$q$ (\citet{Shaked_Shanthikumar_2007(book)}, Theorem 3.A.10).

Index $\xi$ by distribution $J$ and consider first the highest attainable
buyer payoff under $J$. Using $J(\tau_{B}^{J*})=L/H$ and (\ref{eq:xi-identity}),
\[
\xi_{J}(\tau_{B}^{J*})=\frac{H}{\gamma}\,\mathbb{E}[c]\left[\frac{L}{H}-\mathcal{L}_{J}\!\left(\frac{L}{H}\right)\right].
\]
Since $\mathcal{L}_{F}(L/H)\le\mathcal{L}_{G}(L/H)$, we obtain $\xi_{F}(\tau_{B}^{F*})\ge\xi_{G}(\tau_{B}^{G*})$.

Next, fix $u_{B}\in[0,\xi_{G}(\tau_{B}^{G*})]$ and let $\tau^{J}:=\xi_{J}^{-1}(u_{B})$
denote the frontier cutoff under $J$, so that $\bar{U}_{S}^{J}(u_{B})=J(\tau^{J})H$.
By (\ref{eq:xi-identity}), $\xi_{J}(\tau^{J})=u_{B}$ is equivalent
to 
\[
\mathcal{L}_{J}\big(J(\tau^{J})\big)=\frac{L\,\mathbb{E}[c]-\gamma u_{B}}{H\,\mathbb{E}[c]},
\]
whose right-hand side does not depend on $J$. Hence $\mathcal{L}_{F}\big(F(\tau^{F})\big)=\mathcal{L}_{G}\big(G(\tau^{G})\big)\ge\mathcal{L}_{F}\big(G(\tau^{G})\big)$.
Since $\mathcal{L}_{F}$ is strictly increasing, $F(\tau^{F})\ge G(\tau^{G})$,
and therefore $\bar{U}_{S}^{F}(u_{B})=F(\tau^{F})H\ge G(\tau^{G})H=\bar{U}_{S}^{G}(u_{B})$. 
\end{proof}

\subsection*{Proof of Lemma \ref{Lemma:MultiInvestment}}
\begin{proof}
The first condition in Assumption \ref{Assumption} implies that for
any $j<n$,
\[
v_{j}\le\frac{j}{n}(v_{n}-L)+L=\frac{j}{n}v_{n}+\left(\frac{n-j}{n}\right)L,
\]
so
\begin{align*}
v_{n}-v_{j}\ge & v_{n}-\frac{j}{n}v_{n}-\left(\frac{n-j}{n}\right)L=\left(v_{n}-L\right)\left(\frac{n-j}{n}\right)>\left(n-j\right)\bar{c},
\end{align*}
where the last inequality is the second condition in Assumption \ref{Assumption}.
This implies 
\[
\frac{v_{n}-v_{j}}{n-j}>\bar{c}\ge c\ \ \ \forall c\in\left[\underline{c},\bar{c}\right].
\]
\end{proof}
\newpage

\bibliographystyle{econ}
\bibliography{Ref}

\newpage
\setcounter{page}{1}
\renewcommand{\thepage}{OA-\arabic{page}} 
\begin{center} 

{\Large \bf 
Online Appendix: \\
Supplementary Analysis}

\bigskip

\end{center}

\section{\protect\label{OA:ParetoNSC}Additional Analysis for the Pareto Frontier
of Attainable Payoffs}

We provide necessary and sufficient conditions for attaining the payoffs
on the Pareto frontier. Without loss of generality, we assume that
$\mathbb{S}=\left\{ L,H\right\} $ and restrict attention to equilibria
in which the seller sets $p=s$. The following characterizes the set
of all equilibria that attain each payoff pair on the Pareto frontier
\begin{prop}
\label{Proposition:OA-ParetoNSC}Fix a Pareto-frontier payoff pair
$(u_{B},u_{S})$, and let $\tau$ be the unique cutoff satisfying
$SW(\tau)-u_{B}=u_{S}=F(\tau)H$.\footnote{Existence of such a $\tau$ follows from Lemma \ref{Lemma:tau*}.}
An equilibrium under signal structure $\pi$ with the seller offering
$p=s$ attains $(u_{B},u_{S})$ if and only if the following conditions
hold: 
\begin{itemize}
\item the buyer invests if and only if $c\leq\tau$;
\item $\pi_{L}\left(c\right)=c/\gamma$ for all $c>\tau$; 
\item $\pi_{L}\left(c\right)\ge c/\gamma$ for all $c\le\tau$;
\item $\int_{\underline{c}}^{\tau}\left(\pi_{L}\left(c\right)\gamma-c\right)f\left(c\right)dc=u_{B}$.
\end{itemize}
\end{prop}
\begin{proof}
Fix a Pareto-frontier pair $(u_{B},u_{S})$, and let $\tau$ be the
unique cutoff satisfying $SW(\tau)-u_{B}=u_{S}=F(\tau)H$.

We first prove necessity of each condition. Define $\hat{\tau}$ such
that $F\left(\hat{\tau}\right)$ is the probability that the buyer
invests. The necessity part of Proposition \ref{Proposition:AC} implies
$u_{B}+u_{S}\leq SW(\hat{\tau})$ and $u_{S}\geq F(\hat{\tau})H$.
Since $u_{B}+u_{S}=SW(\tau)$ and $u_{S}=F(\tau)H$, and since both
$SW$ and $F$ are strictly increasing, these inequalities imply $\hat{\tau}=\tau$.
Since $u_{B}+u_{S}=SW(\tau)$, the investing set must be the lowest-cost
set of measure $F(\tau)$. Therefore the buyer must invest if and
only if $c\leq\tau$.

Given the cutoff strategy and $p=s$, type $c$ faces $p=L$ with
probability $\pi_{L}\left(c\right)$. Buyer optimality in (\ref{eq:BIC})
requires $\pi_{L}\left(c\right)\geq c/\gamma$ for $c\leq\tau$ and
$\pi_{L}\left(c\right)\leq c/\gamma$ for $c>\tau$. The surplus equality
established above forces $\pi_{L}\left(c\right)=c/\gamma$ for almost
every $c>\tau$. The buyer's payoff is $\int_{\underline{c}}^{\tau}\left(\pi_{L}\left(c\right)\gamma-c\right)f\left(c\right)dc$,
which must equal $u_{B}$.

Next, we establish sufficiency. Suppose the buyer invests if and only
if $c\leq\tau$, the seller sets $p=s$, and $\pi$ satisfies the
stated conditions. The buyer incentive constraints hold because types
$c\leq\tau$ obtain weakly nonnegative payoff from investing, while
types $c>\tau$ are indifferent between investing and not investing. 

The buyer payoff is $u_{B}$ by construction. The total surplus is
$SW(\tau)$: investing types generate $H-c$, and noninvesting types
trade with probability $c/\gamma$. Hence the seller payoff is $SW(\tau)-u_{B}=F(\tau)H=u_{S}$.

It remains only to verify seller obedience. Let $Q_{s}$ be the probability
of signal $s$, and $Q_{s}^{I}$ the joint probability of signal $s$
and investment. The seller obeys signal $L$ if $Q_{L}L-Q_{L}^{I}H\geq0$.
Using $Q_{L}^{I}+Q_{H}^{I}=F(\tau)$, we have
\[
Q_{L}L-Q_{L}^{I}H=Q_{L}L+Q_{H}^{I}H-F(\tau)H=u_{S}-F(\tau)H=0.
\]
The seller obeys signal $H$ if $Q_{H}^{I}H-Q_{H}L\geq0$. Using $Q_{L}+Q_{H}=1$,
we have
\[
Q_{H}^{I}H-Q_{H}L=Q_{H}^{I}H+Q_{L}L-L=u_{S}-L\geq0.
\]
\end{proof}
Proposition \ref{Proposition:OA-ParetoNSC} formalizes the observation
about $\pi_{L}\left(c\right)=c/\gamma$ for noninvesting types, as
discussed after Proposition \ref{Proposition:Construction} in the
main text. Seller obedience requires enough accuracy in distinguishing
investing from noninvesting types. Thus profiling error cannot be
wasted on noninvesting types: they receive zero payoff anyway, and
any additional error assigned to them only weakens the information
that sustains seller obedience.

The remaining freedom is entirely on the investing side. For each
investing type, equilibrium incentive compatibility only requires
$\pi_{L}(c)\geq c/\gamma$, and the aggregate amount of such error
must deliver buyer payoff $u_{B}$. The dual-cutoff construction used
in the main text is therefore one convenient way to allocate this
error: it gives all surplus to the lowest-cost investing types up
to $\phi$ and leaves the other investing types indifferent. Proposition
\ref{Proposition:OA-ParetoNSC} shows that this allocation is not
essential. Any distribution of profiling error across investing types
that satisfies the same aggregate buyer-payoff condition attains the
same Pareto-frontier payoff.

This is also the sense in which buyer protection comes from targeted
ignorance rather than from seller ignorance per se. A profiling ban
spreads error across all types and therefore spends some of it on
noninvestors, where it has no buyer-welfare value. A frontier equilibrium
instead pairs maximal accuracy about noninvesting types with enough
mis-profiling of investing types to let the buyer retain the desired
surplus. The proposition therefore clarifies the statement in the
main text: information provision to the seller can coexist with buyer
protection, but only when the residual uncertainty is concentrated
on the types whose investment surplus is at stake.

\section{\protect\label{OA:Multiple-Investments}Additional Analysis for Section
\ref{Section:Multiple-Investments}}

\subsection*{When the efficient option is not the top option}

Section \ref{Section:Multiple-Investments} illustrates how our results
hold when the buyer has multiple investment options. The underlying
assumption is that the efficient investment option is the highest
option $n$. We show that this assumption is also not necessary.

Without loss of generality, let 
\[
k:=\underset{i\in\left\{ 1,\dots,n\right\} }{\arg\max}\frac{v_{i}-L}{i},
\]
so the investment with the highest average return is option $k$,
which is possibly lower than $n$.
\begin{assumption}
\label{Assumption:2}$v_{k}-k\bar{c}>L$ and $v_{k}-kc>v_{i}-ic$
for all $i\ne k$ and $c\in\left[\underline{c},\bar{c}\right]$.
\end{assumption}
Assumption \ref{Assumption:2} says that option $k$ is the efficient
option for every type.
\begin{prop}
Under Assumption \ref{Assumption:2}, an IR pair $(u_{B},u_{S})$
is attainable if and only if there exists $\tau\in[\underline{c},\bar{c}]$
such that $SW_{k}(\tau)-u_{B}\ge u_{S}\ge F(\tau)v_{k}$.
\end{prop}
\begin{proof}
For necessity, fix an equilibrium attaining $(u_{B},u_{S})$. Let
$x(c)$ denote the probability that type $c$ chooses an option $j\ge k$,
and choose $\tau\in[\underline{c},\bar{c}]$ such that $F(\tau)=\int_{\underline{c}}^{\bar{c}}x(c)f(c)\,dc$.
By deviating to $p=v_{k}$, the seller sells whenever the buyer chooses
$j\ge k$. Seller optimality therefore implies
\[
u_{S}\ge v_{k}\int_{\underline{c}}^{\bar{c}}x(c)f(c)\,dc=F(\tau)v_{k}.
\]
We next bound total surplus. If type $c$ chooses $j\ge k$, the resulting
surplus is at most $v_{j}-jc\le v_{k}-kc$, where the inequality follows
from Assumption \ref{Assumption:2}. Now suppose that type $c$ chooses
$j<k$, and let $q_{j}(c)$ be the probability that trade occurs.
Deviating from $j$ to $k$ would increase the buyer's payoff by at
least $q_{j}(c)(v_{k}-v_{j})-c(k-j)$. Buyer optimality thus requires
\[
q_{j}(c)\le c\frac{k-j}{v_{k}-v_{j}}.
\]
Because $k$ has the highest average return, $v_{j}\le L+\frac{j}{k}(v_{k}-L)$.
Consequently, as in the proof of Proposition \ref{Proposition:Multiple-Investment},
\[
q_{j}(c)v_{j}-jc\le c\left[\frac{k-j}{v_{k}-v_{j}}v_{j}-j\right]\le\frac{kc}{v_{k}-L}L.
\]
This argument includes $j=0$. Hence, writing $s(c)$ for the expected
surplus generated by type $c$,
\[
s(c)\le x(c)(v_{k}-kc)+[1-x(c)]\frac{kc}{v_{k}-L}L.
\]
The difference between the two bounds,
\[
(v_{k}-kc)-\frac{kc}{v_{k}-L}L=v_{k}\left(1-\frac{kc}{v_{k}-L}\right),
\]
is decreasing in $c$. Given that $\int x(c)f(c)\,dc=F(\tau)$, total
surplus is therefore maximized when $x(c)=1$ for $c\le\tau$ and
$x(c)=0$ for $c>\tau$. It follows that
\[
u_{B}+u_{S}\le\int_{\underline{c}}^{\tau}(v_{k}-kc)f(c)\,dc+\int_{\tau}^{\bar{c}}\frac{kc}{v_{k}-L}Lf(c)\,dc=SW_{k}(\tau).
\]
Thus $SW_{k}(\tau)-u_{B}\ge u_{S}$.

For sufficiency, use the dual-cutoff construction from Proposition
\ref{Proposition:Multiple-Investment}, replacing $n$ by $k$. The
proof of Proposition \ref{Proposition:Multiple-Investment} verifies
seller obedience and rules out deviations to every $j<k$. It remains
to consider $j>k$. Under the construction, prices lie in $\{L,v_{k}\}$,
so both options $j$ and $k$ trade following either signal. The payoff
difference from choosing $j$ rather than $k$ is therefore $(v_{j}-jc)-(v_{k}-kc)<0$
by Assumption \ref{Assumption:2}. No deviation to $j>k$ is profitable.
The construction is consequently an equilibrium and attains $(u_{B},u_{S})$.

\end{proof}

\end{document}